\documentclass[12pt]{article}

\usepackage{hyperref}
\usepackage{amsthm}
\usepackage{amsmath}
\usepackage{amssymb}
\usepackage{amsfonts}
\usepackage{graphicx}
\usepackage{enumerate}
\usepackage{tikz}
\usetikzlibrary{automata, arrows}

\newtheorem{theorem}{Theorem}[section]
\newtheorem{lemma}[theorem]{Lemma}

\newtheorem{problem}[theorem]{Problem}
\newtheorem{conjecture}[theorem]{Conjecture}

\usepackage{complexity}

\newcommand{\tpart}{\textsc{3-partition}}
\newcommand{\mfst}{\textsc{Encoding by FST}}
\newcommand{\mefst}{\textsc{Promise encoding by FST}}
\newcommand{\mmefst}{\textsc{Modified promise encoding by FST}}
\newcommand{\pats}{\textsc{PATS}}

\newcommand{\north}{\ensuremath{\mathtt{N}}}
\newcommand{\west}{\ensuremath{\mathtt{W}}}
\newcommand{\south}{\ensuremath{\mathtt{S}}}
\newcommand{\east}{\ensuremath{\mathtt{E}}}

\newcommand{\Tile}[7]{\draw[fill=#3] (#1, #2)++(0.45,0.45) -- node[above] {#4} ++(180:0.9) -- node[left] {#5} ++(270:0.9) -- node[below] {#6} ++(0:0.9) -- node[right] {#7} ++(90:0.9);}
\newcommand{\Tzero}[6]{\Tile{#1}{#2}{cyan}{#3}{#4}{#5}{#6};}
\newcommand{\Tone}[6]{\Tile{#1}{#2}{gray}{#3}{#4}{#5}{#6};}
\newcommand{\Ttwo}[6]{\Tile{#1}{#2}{orange}{#3}{#4}{#5}{#6};}

\author{Shinnosuke Seki
\thanks{The University of Electro-Communications, Tokyo, Japan, \protect{\texttt{s.seki@uec.ac.jp}}}
\thanks{Work supported in part by JST Program to Disseminate Tenure Tracking System, MEXT, Japan, No.~6F36 and by JSPS Grant-in-Aid for Research Activity Start-up No.~15H06212.}
\and Andrew Winslow\thanks{University of Texas Rio Grande Valley, Edinburg, TX, USA, \protect{\texttt{andrew.winslow@utrgv.edu}}}}
\title{The Complexity of Fixed-Height\\Patterned Tile Self-Assembly}
\date{}

\begin{document}

\maketitle


\begin{abstract}
We characterize the complexity of the \pats{} problem for patterns of fixed height and color count in variants of the model where seed glues are either chosen or fixed and identical (so-called \emph{non-uniform} and \emph{uniform} variants).
We prove that both variants are \NP-complete for patterns of height~2 or more and admit $O(n)$-time algorithms for patterns of height~1.
We also prove that if the height and number of colors in the pattern is fixed, the non-uniform variant admits a $O(n)$-time algorithm while the uniform variant remains \NP-complete. 
The \NP-completeness results use a new reduction from a constrained version of a problem on finite state transducers.
\end{abstract}

\section{Introduction}
Winfree~\cite{Winfree-1998a} introduced the \emph{abstract tile assembly model (aTAM)} to capture nanoscale systems of DNA-based particles aggregating to form intricate crystals, leading to an entire field devoted to understanding the theoretical limits of such systems (see surveys by Doty~\cite{Doty-2012a} and Patitz~\cite{Patitz-2012a}).
Ma and Lombardi~\cite{Ma-2008a} introduced the \emph{patterned self-assembly tile set synthesis (\pats{})} problem, of designing a tile set of minimum size that assembles into a given $n \times h$ colored pattern by attaching to an L-shaped seed.

Czeizler and Popa~\cite{Czeizler-2013a} were the first to provide a proof that the \pats{} problem is \NP-hard, thus establishing the problem as \NP-complete.
Subsequent work studied the hardness of the constrained version where the patterns have at most $c$ colors, called the \emph{$c$-\pats{}} problem. 
This line of work proved the 60-\pats{}~\cite{Seki-2013a}, 29-\pats{}~\cite{Johnsen-2013a}, 11-\pats{}~\cite{Johnsen-2015a}, and finally the 2-\pats{}~\cite{Kari-2015a} problems \NP-complete.

Here we study the complexity of parameterized \emph{height-$h$} \pats{} and $c$-\pats{} problems where patterns have a specified fixed height $h$ and increasing width $n$.
We consider both \emph{uniform} and \emph{non-uniform} model variants, where the glues along the seed are fixed and identical or chosen in tandem with the tile set, respectively.
We characterize the computational complexity of these problems via the following results:
\begin{enumerate}[$\bullet$]
\item The height-$2$ \pats{} problem is \NP-complete in both models (Sec.~\ref{sec:PATS-hard}).
\item The uniform height-$2$ $3$-\pats{} problem is \NP-complete (Sec.~\ref{sec:cPATS-hard}). 
\item The non-uniform height-$h$ $c$-\pats{} problem and uniform height-1 \pats{} problems admit $c^{c^{O(h)}}n$-time and $O(n)$-time algorithms, respectively (Sec.~\ref{sec:PATS-easy}).
\end{enumerate}
The \NP-completeness results also apply to patterns of height greater than~2.
Thus the complexity of the \pats{} problem for all combinations of height, color, and uniformity are characterized, except uniform height-2 2-\pats{}.

The \NP-hardness reductions are based on a reduction for a new variant of the minimum-state finite state transducer problem, originally proved \NP-hard by Angluin~\cite{Angluin-1978a} and by Vazirani and Vazirani~\cite{Vazirani-1983a}.
In this variant, any solution transducer is also promised to satisfy additional constraints on its transitions.
The reduction is also substantially simpler than the reduction given in~\cite{Angluin-1978a} and uses input and output strings of just two symbols, rather than the three of~\cite{Vazirani-1983a}.

\section{Preliminaries}

\paragraph{Patterns, tiles, assemblies, and seeds.}
Define $\mathbb{N}_k = \{1, 2, \dots, k\}$.
A \emph{pattern} is a partial function $P : \mathbb{N}^2 \rightarrow C$, i.e. a function that maps a rectangular region of lattice points to a set of colors $C$ (see Figure~\ref{fig:mintile-examples}).
If $\mathrm{dom}(P) = \mathbb{N}_w \times \mathbb{N}_h$, then $P$ is a \emph{width-$w$ height-$h$} pattern.
The codomain of $P$, i.e. the colors seen in the pattern, is denoted $\mathrm{color}(P)$. 
A pattern $P$ is \textit{$c$-color} provided $|\mathrm{color}(P)| \le c$. 

A \emph{tile type} $t$ is a colored unit square with each edge labeled; these labels are called \emph{glues}.
A tile type's color is denoted $\mathrm{color}(t)$.
For a direction $d \in \{\north, \west, \south, \east\}$, $t[d]$ denotes the glue assigned to side $d$ of $t$. 
A tile type is non-rotatable, and thus is uniquely identified by its color and four glues.
Instances of tile types, called \emph{tiles}, are placed with their centers in $\mathbb{N}^2$.

An \emph{assembly} is an arrangement of tiles from a set of tile types $T$; formally a partial function $A : \mathbb{N}^2 \rightarrow T \cup \{\varnothing\}$.
A \emph{seed} is an ``L-shaped'' assembly with domain $\{(0, 0)\} \cup \{(x, 0) : x \in \mathbb{N}_w\} \cup \{(0, y) : y \in \mathbb{N}_h\}$ for some $w, h \in \mathbb{N}$.
The \emph{pattern} of an assembly $A$ is defined as $P_A((x, y)) = \mathrm{color}(A((x, y)))$ for $(x, y) \in \mathrm{dom}(A) \cap \mathbb{N}^2$, i.e. the color pattern of $A$, \emph{excluding} the seed.

\paragraph{RTASs.}
A \emph{rectilinear tile assembly system (RTAS)} is a pair $\mathcal{T} = (T, \sigma)$, where $T$ is a set of tile types and $\sigma$ is a \emph{seed}.
An assembly $A$ \emph{yields} an assembly $A'$ with $\mathrm{dom}(A') = \mathrm{dom}(A) \cup \{(x, y)\}$ provided $(x-1, y), (x, y-1) \in \mathrm{dom}(A)$ and $A((x-1, y))[\east] = A'((x, y))[\west]$,  $A((x, y-1))[\north] = A'((x, y))[\south]$. 
The \emph{producible assemblies} of an RTAS are those that can be yielded, starting with the seed assembly $\sigma$.
That is:

\begin{description}
\item[RTAS Tiling Rule:] 
A tile of type $t$ can be added to an assembly $A$ at location $(x, y)$ provided $(x-1, y), (x, y-1) \in \mathrm{dom}(A)$ and the east and north glues of the tiles at $(x-1, y)$ and $(x, y-1)$ are the same as the west and south glues of $t$, respectively. 
\end{description}

As a result, tiling proceeds from southwest to northeast, i.e., a tile is first placed at $(1, 1)$, then at either $(1, 2)$ or $(2, 1)$, etc.
The \emph{terminal assemblies} of an RTAS are the producible assemblies that do not yield other (larger) assemblies.
If every terminal assembly of the system has pattern $P$, the system is said to \emph{uniquely self-assemble $P$}.
An RTAS $(T, \sigma)$ is \emph{directed}, i.e. deterministic, provided for any distinct tile types $t_1, t_2 \in T$, either $t_1[\west] \neq t_2[\west]$ or $t_1[\south] \neq t_2[\south]$. 

\paragraph{Uniform RTASs.}
We also define a practical variant of an RTAS called a \emph{uniform} RTAS. 
An RTAS $(T, \sigma)$ is \textit{uniform} provided there exist two glues $\ell_\east, \ell_\north$ such that $\sigma((x, 0))[\east] = \ell_\east$ for all $x \in \mathbb{N}_w$ and $\sigma((0, y))[\north] = \ell_\north$ for all $y \in \mathbb{N}_h$.
In other words, the seed glues cannot be programmed and are generic.

\paragraph{The PATS problem.}
The \emph{pattern self-assembly tile set synthesis problem (PATS)}~\cite{Ma-2008a} asks for the minimum-size RTAS that uniquely self-assembles a given rectangular color pattern, where the size of an RTAS $(T, \sigma)$ is $|T|$, the number of tile types.
Bounding the number of colors or height of the input pattern yields the following practically motivated special cases of PATS: 

\begin{figure}[ht]
\centering
\includegraphics[scale=0.7]{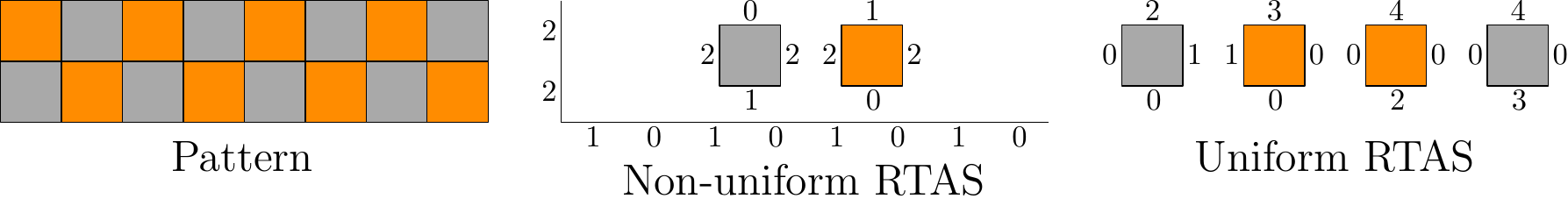}
\caption{A height-2 2-color pattern and minimum-size RTASs uniquely assembling the pattern in the uniform and non-uniform models.}
\label{fig:mintile-examples}
\end{figure}

\begin{problem}[$c$-colored \pats{} or $c$-\pats{}]
Given a $c$-colored pattern $P$ and integer $t$, does there exist an RTAS of size $\leq t$ that uniquely self-assembles $P$?
\end{problem}

\begin{problem}[Height-$h$ \pats{}]
Given a height-$h$ pattern $P$ and integer $t$, does there exist an RTAS of size $\leq t$ that uniquely self-assembles $P$?
\end{problem}

Restricting the system to be uniform gives rise to \textit{uniform} variants as well, contrasting with the conventional \textit{non-uniform} variants.

\section{Minimum-State Finite State Transducer is \NP-hard}
\label{sec:FST-reduction}

We start with a reduction from a well-known \NP-complete problem on integers, called \tpart{}, to a problem on \emph{finite state transducers} or \emph{FSTs}: deterministic finite automata where each transition is augmented with an output symbol and thus \emph{transduces} an input string into an output string of the same length.

\begin{problem}[\tpart{}]
Given a set of integers $A = \{a_1, a_2, \dots, a_{3n}\}$ with $\sum_{a_i \in A}{a_i}/n = p$ and $p/4 < a_i < p/2$, does there exist a partition of $A$ into $n$ sets, each with sum $p$?
\end{problem}

\begin{theorem}[\cite{Garey-1975a}]
\tpart{} is strongly \NP-hard, i.e., is \NP-hard when the elements of $A$ are given in unary.
\end{theorem}

Formally, a FST is a 4-tuple $T = \langle \Sigma, Q, s_0, \delta \rangle$, where $\Sigma$ is the \emph{alphabet}, $Q$ is a finite set of \emph{states} of $T$, $s_0 \in Q$ is the \emph{start state} of $T$, and $\delta : Q \times \Sigma \rightarrow Q \times \Sigma$ is the \emph{transition function} of $T$.
An input-output quadruple $\delta(s_i, b) = (s_j, b')$ is a \emph{transition}, specifically a \emph{$(b, b')$-transition} or a \emph{$b$-transition}. 
The size of $T$ is equal to $|Q|$.

\begin{problem}[\mfst{}]
Given two strings $S$, $S'$ and integer $K$, does there exist a FST with at most $K$ states that transduces $S$ to $S'$?
\end{problem}

The \mfst{} problem was previously shown to be \NP-hard by Angluin~\cite{Angluin-1978a} and by Vazirani and Vazirani~\cite{Vazirani-1983a}.
Here we prove a constrained variant is also \NP-hard:

\begin{problem}[\mefst{}]
Given two strings $S$, $S'$ and an integer $K$ with the following promises about any FST $T$ with at most $K$ states transducing $S$ to $S'$, does such a $T$ exist?
\begin{itemize}
\item Each state of $T$ has exactly one incoming 0-transition.
\item Each state of $T$ has exactly one incoming 1-transition.
\item When transducing $S$ to $S'$:
\begin{itemize}
\item $K-1$ distinct $(0, 0)$-transitions are used.
\item $K$ distinct $(1, 1)$-transitions are used.
\item 1 distinct $(0, 1)$-transition is used.
\item The order that transitions are traversed is given as part of the input.\footnote{The order is given by a sequence $t_1, t_2, \dots, t_{|S|}$ with each $t_i \in \{1, 2, \dots, 2K\}$ specifying which transition is traversed. Thus the order that transitions are revisited is specified, but the states visited are not.}
\end{itemize}
\end{itemize}
\end{problem}

\begin{lemma}
\label{lem:promise-min-state-FST-NP-hard}
The \mefst{} problem is \NP-hard.
\end{lemma}

\begin{proof}
The reduction is initially presented as a reduction from \tpart{} to \mfst{}.
The final paragraph extends the reduction to \mefst{}.

We borrow from~\cite{Vazirani-1983a} the approach of constructing $S$ and $S'$ by concatenating \emph{segments}: pairs of input and output substrings of equal length that enforce specific structure in a solution FST.
An input string $X$ and output string $Y$ paired as a segment is denoted $X \rightarrow Y$.

The integer output by the reduction is $K = 3pn + n + 1$, where $n$ is the number of parts in the partition and $p$ the size of each part. 
The first segment is $0^{K-1} 0 0^{K-1} \rightarrow 0^{K-1} 1 0^{K-1}$.
This segment enforces that a solution FST must have $K$ states; label them $s_1, s_2, \dots, s_K$.
Then for all $i < n$, $\delta(s_i, 0) = (s_{i+1}, 0)$ and $\delta(s_K, 0) = (s_1, 1)$.

\begin{figure}[ht]
\centering
\includegraphics[scale=0.83]{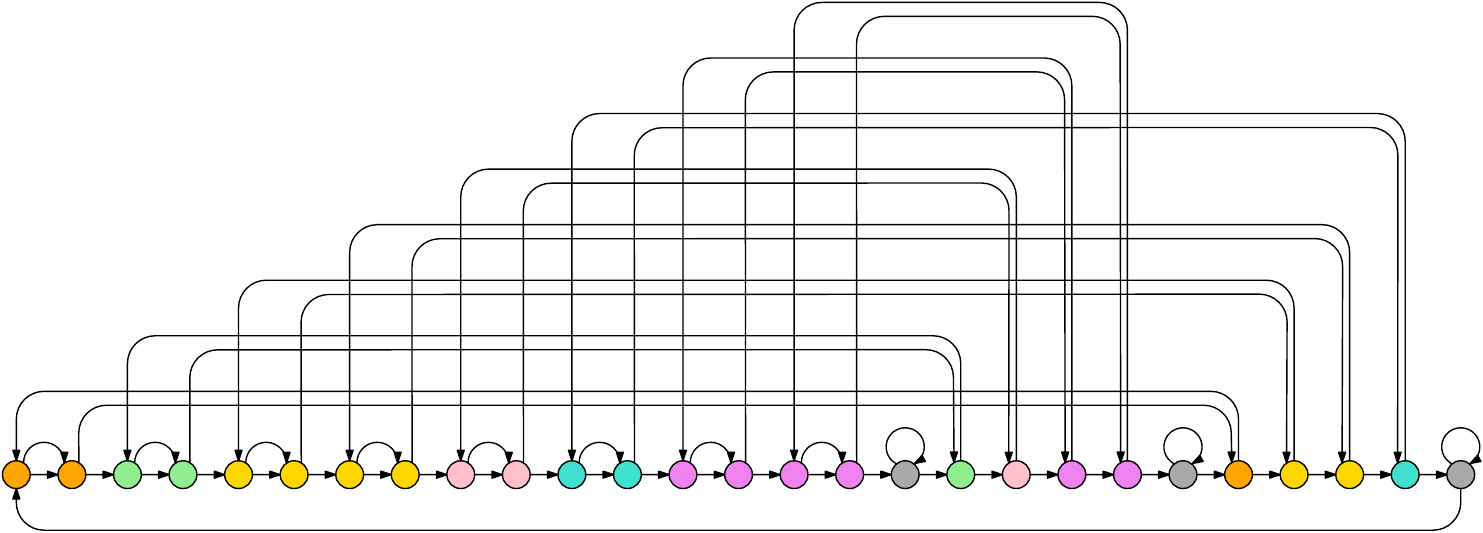}
\caption{A solution FST for a toy reduction from \tpart{} to \mefst{} with integers $a_1 = 1, a_2 = 1, a_3 = 2, a_4 = 1, a_5 = 1, a_6 = 2$ (invalid due to duplicate elements and violating $p/4 < a_i < p/2$, but used for illustrative purposes).
The left-to-right states are $s_1$ to $s_K$, colored by their half-fixed interval.
Transitions above the states are $(1, 1)$-transitions.
All others are $(0, 0)$-transitions except the lowermost, a $(0, 1)$-transition.
The solution corresponds to the partition $\{\{a_2, a_4, a_6\}, \{a_1, a_3, a_5\}\}$, encoded in the assignment of the right halves of half-fixed intervals to states between the rightmost three fixed singletons (gray).}
\label{fig:fst-reduction}
\end{figure}

The problem of partitioning integers of $A$ into sets of size $p$ is implemented in the collection of 1-transitions that leave states.
Each state has a 1-transition that either points to itself (a \emph{fixed singleton}) or is one edge in a 3-cycle formed by two consecutive specified states and an unspecified third state (a \emph{half-fixed triple}).
Half-fixed triples are further organized into \emph{half-fixed intervals}, each consisting of a set of $2a_i$ consecutive specified states and a set of $a_i$ consecutive unspecified states for some $a_i$.
The states are partitioned into three subsets: 
\begin{enumerate}
\item States $s_1$ through $s_{2pn}$ consist of the specified halves of the half-fixed intervals ($2a_i$ consecutive states for each integer $a_i \in A$).
\item $n+1$ equally-spaced fixed singletons in states $s_{2pn+1}, \dots, s_K$ (partitioning the remaining $pn$ states into $n$ sets of $p$ consecutive states each).
\item The remaining $pn$ states in states $s_{2pn+1}, \dots, s_K$ partitioned into $n$ sets of $p$ consecutive states. 
\end{enumerate}
See Figure~\ref{fig:fst-reduction} for a toy example of the reduction.

The reduction works by leaving one aspect of the FST unspecified: the assignment of the unspecified halves of half-fixed intervals (sequences of $a_i$ consecutive states for each $a_i \in A$) to the $n$ sequences of $p$ consecutive states each (the third subset mentioned above).

These unspecified halves must be assigned to these sequences of $p$ consecutive states, since all other states are either fixed halves of half-fixed intervals (first subset above) or fixed singletons (second subset above).
Also, since no state can lie on multiple 3-cycles of 1-transitions, no state can be in multiple half-fixed triples.
So no state can be in multiple half-fixed intervals and thus the unspecified halves of half-fixed intervals cannot overlap.

Thus the unspecified halves of half-fixed intervals may be viewed as indivisible ``rods'' of consecutive states (of lengths $a_i$ for each $a_i \in A$) that must be placed into $n$ state ``boxes'', each of length $p$, without overlap. 
Since the sum of the rod lengths is exactly $pn$, this problem is equivalent to partitioning the rods into $n$ subsets, each of sum $p$, i.e. \tpart{}.

In Figure~2, the input (invalid) \tpart{} instance $A = \{a_1=1, a_2=1, a_3=2, a_4=1, a_5=1, a_6=2\}$ has a solution $\{\{a_2, a_4, a_6\}, \{a_1, a_3, a_5\}\}$.
This corresponds to a solution FST obtained by assigning the unspecified halves of half-fixed intervals corresponding to $a_2$, $a_4$, and $a_6$ (of lengths 1, 1, and 2, respectively) into the first sequence of~$p=4$ consecutive states, and similarly the unspecified halves of the half-fixed intervals corresponding to $a_1$, $a_3$, and $a_5$ (of lengths 1, 2, and 1, resp.) into the second sequence of $p=4$ consecutive states. 
All that remains is to describe the segments that force the construction of a fixed singleton, half-fixed triple, and half-fixed interval.

\paragraph{Fixed singleton.} 
The fixed singleton segment ensures that a given state $s_i$ has $\delta(s_i, 1) = (s_i, 1)$.
This is done by moving the current state to $s_i$, transducing a~1 to a~1, and checking whether the current state is still $s_i$ (see Figure~\ref{fig:fst-fixed-pair}).

\begin{figure}[ht!]
\centering
\includegraphics[scale=1.0]{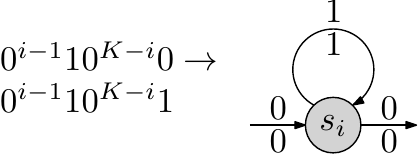}
\caption{The fixed single segment and corresponding FST structure enforced.}
\label{fig:fst-fixed-pair}
\end{figure}

\paragraph{Half-fixed triple.}
The half-fixed triple segment forces two specified \emph{fixed} states $s_i$, $s_{i+1}$ and an unspecified \emph{free} third state $s_j$ to have $\delta(s_i, 1) = (s_{i+1}, 1)$, $\delta(s_{i+1}, 1) = (s_j, 1)$, and $\delta(s_j, 1) = (s_i, 1)$ (see Figure~\ref{fig:fst-half-fixed-triple}).  

\begin{figure}[ht!]
\centering
\includegraphics[scale=1.0]{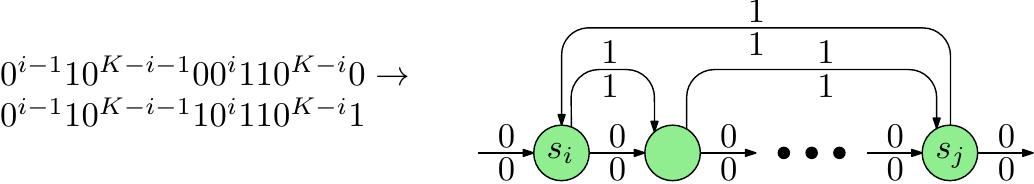}
\caption{The half-fixed triple segment and corresponding FST structure enforced.}
\label{fig:fst-half-fixed-triple}
\end{figure}

The segment consists of two subsegments that each ensures a portion of the structure.
The first, $0^{i-1}10^{K-i-1}0$ $\rightarrow$ $0^{i-1}10^{K-i-1}1$, ensures that $\delta(s_i, 1) = (s_{i+1}, 1)$.
The second, $0^i 11 0^{K-i}0 \rightarrow 0^i 11 0^{K-i}1$, ensures that $\delta(s_{i+1}, 1) = (s_j, 1)$ and $\delta(s_j, 1) = (s_i, 1)$.
The state $s_j$ cannot be in a fixed state of another half-fixed triple segment with fixed states $s_i'$, $s_{i+1}'$ and free state $s_j'$, as then either:
\begin{itemize}
\item $s_j = s_i'$ and thus $\delta(s_j, 1) = (s_{i+1}', 1) \neq (s_i, 1)$ (and thus the segment $0^{i-1}10^{K-i-1}0 \rightarrow 0^{i-1}10^{K-i-1}1$ is not transduced).
\item $s_j = s_{i+1}'$ and $(s_j', 1) = \delta(s_{i+1}', 1) = (s_i, 1)$, so $\delta(s_j', 1) = (s_{i+1}, 1) \neq (s_i', 1)$ (and thus the segment $0^{i'-1}10^{K-i'-1}0 \rightarrow 0^{i'-1}10^{K-i'-1}1$ is not transduced).
\end{itemize}
 
\begin{figure}[ht!]
\centering
\includegraphics[scale=1.0]{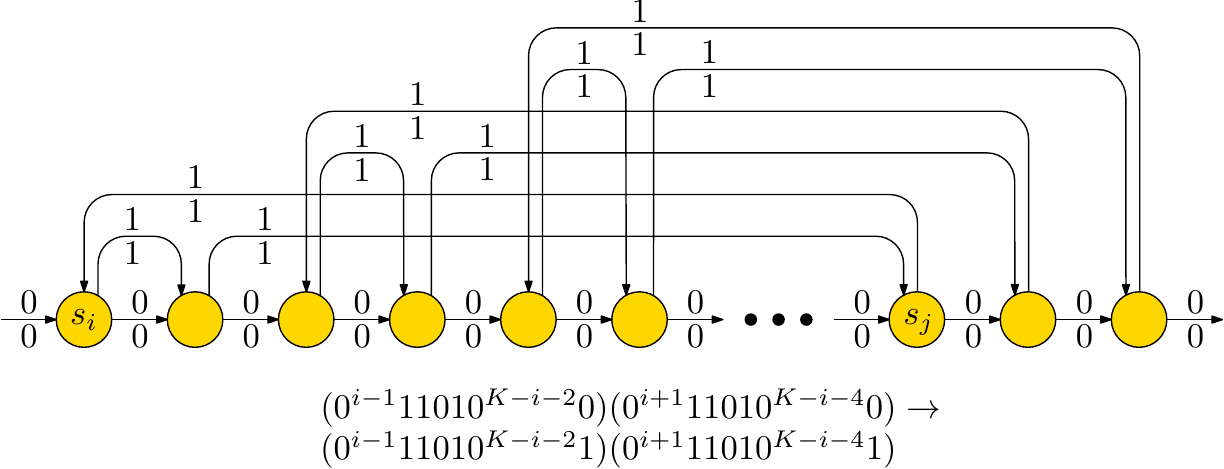}
\caption{The half-fixed interval segment for three consecutive free states and corresponding FST structure enforced.}
\label{fig:fst-half-fixed-interval}
\end{figure}

\paragraph{Half-fixed interval.}
The half-fixed interval forces a collection of half-fixed triples with consecutive fixed states to also have consecutive free states.
It does so by a simple traversal of the free states, checking that each has the expected pair of consecutive fixed states (see Figure~\ref{fig:fst-half-fixed-interval}). 

\paragraph{Promises.}
Any solution transduction for the previous reduction uses an FST where each state has at most one incoming 1-transition and 0-transition, since every transition lies on a cycle (of length 1,~3, or~$K$).
Also, any solution transduction by an FST with $K$ states traverses $2K$ distinct transitions (with $K-1$ $(0, 0)$-transitions, 1 $(0, 1)$-transition, and $K$ $(1, 1)$-transitions).
Any other solution FST must have at least $2K$ states and traverse at least $2K+1$ distinct transitions: $2K$ 0-transitions and at least one 1-transition.
Thus the reduction proves that not only \mfst{}, but also \mefst{} is \NP-hard.
\end{proof}

\section{Height-$2$ \pats{} is \NP-complete}
\label{sec:PATS-hard}

G\"{o}\"{o}s and Orponen~\cite{Goos-2010a} establish that all the variations of the \pats{} problem considered here are in \NP.
So we need only consider their \NP-hardness.

\begin{theorem}
\label{thm:2xn-NU-PATS-NP-hard}
The non-uniform height-2 \pats{} problem is \NP-hard.
\end{theorem}

\begin{proof}
The pattern output by the reduction consists of a bottom row encoding $S$ and a top row encoding the sequence of transitions traversed when transducing $S$ to $S'$ (provided as part of the \mefst{} instance).
The bottom row encoding uses two colors, pink and red, corresponding to the two symbols in $S$.
The top row encoding uses $2K$ colors, one for each transition used in the transduction of $S$ to $S'$.
The number of tile types permitted is $T = 2K+2$: one type per color.

The north glues of the bottom row either encode $S$ (distinct north glues for the pink and red tile types) or $0^{|S|}$ (same glue).
The latter is impossible, since then the leftmost $|S|$ locations of the top row are filled by many repetitions of the same $K$ transitions. 
So the north glues of the bottom row encode $S$.

A set of $2K$ tile types that assemble the top row is equivalent to a set of $2K$ transitions transducing $S$ to $S'$, with source and destination states corresponding to west and east glues. 
So the top row can be assembled using $2K$ tile types exactly when $S$ can be transduced to $S'$ using $2K$ transitions of the specified types traversed in the specified order. 
Note that these tile types may share the east glue in common, corresponding to that a state has more than one incoming 0-transition or more than one incoming 1-transition. 
Nevertheless, we can say that the pattern can be assembled using a tile set of at most $2K$ types exactly when there is a solution transducer for the corresponding instance of a harder variant of \mefst{} that is obtained by dropping the first two promises, on incoming transitions, from \mefst{} (the harder variant remains \NP-hard). 
\end{proof}

\begin{theorem}
\label{thm:2xn-U-PATS-NP-hard}
The uniform height-$2$ \pats{} problem is \NP-hard.
\end{theorem}

\begin{proof}
The pattern output by the reduction is the following, and consists of a left \emph{input half} and right \emph{transduction half}:
\[\arraycolsep=1.4pt
\begin{array}{ccc}
S & 
\square &
{\rm (transduction)} \\

\underbrace{\textcolor{orange}{\blacksquare} \textcolor{orange}{\blacksquare} \textcolor{orange}{\blacksquare} \cdots \textcolor{orange}{\blacksquare} \textcolor{orange}{\blacksquare} \textcolor{orange}{\blacksquare}}_{|S|} & 
\textcolor{black}{\blacksquare} &
\underbrace{\textcolor{orange}{\blacksquare} \textcolor{orange}{\blacksquare} \textcolor{orange}{\blacksquare} \cdots \textcolor{orange}{\blacksquare} \textcolor{orange}{\blacksquare} \textcolor{orange}{\blacksquare}}_{|S|}  
\end{array}
\]

The color patterns in the top rows of the input and transduction halves are identical to the bottom and top rows of the pattern used in the proof of Theorem~\ref{thm:2xn-NU-PATS-NP-hard}, respectively.
That is, they consist of two colors encoding $S$, and $2K$ colors encoding the sequence of transitions.
 
The bottom row of the input half consists of $|S|$ orange tiles; these must be of $|S|$ distinct types, otherwise the entire bottom row consists of orange tiles.
The number of tile types output by the reduction is $T = |S|+2K+4$, thus any solution set of tile types has exactly the following tile types: 
\begin{itemize}
\item $|S|$ orange tile types.
\item 2 types used in the top row of the input half.
\item $2K$ types used in the top row of the transduction half.
\item White and black types used for the two tiles between the halves.
\end{itemize}

Since only one tile type per color is used in the top row of the input half, the north glues of the length-$|S|$ orange tile sequence must encode $S$.
Also, since the bottom row of the transduction half consists of a length-$|S|$ sequence of orange tiles and these tiles are all contained in the bottom row of the input half, they must be the same sequence of tiles.
So the north glues of the bottom row of the transduction half encode $S$, and thus any set of $2K$ tile types used in the top row of the right half corresponds exactly to a set of transitions in a solution FST for the \mefst{} instance.
\end{proof}

The addition of more rows with a new common color and increasing $T$ by $1$ suffices to prove both the uniform and non-uniform variants \NP-hard for greater heights.

\section{Uniform Height-$2$ $3$-\pats{} is \NP-complete}
\label{sec:cPATS-hard}

The next problem is a variation of \mefst{}, which we prove to be \NP-complete by modifying the reduction from \tpart{} to \mefst{}.
This reduction serves as an introduction and warm-up for the subsequent reduction from \tpart{} to \mmefst{}. 

\begin{problem}[\mmefst{}]
Given two strings $S$, $S'$ and an integer $K \not\equiv 0 \pmod{3}$ with the following promises about any FST $T$ with at most $K$ states transducing $S$ to $S'$, does such a $T$ exist?
\begin{itemize}
\item Each state of $T$ has exactly one incoming 0-transition.
\item Each state of $T$ has exactly one incoming 1-transition.
\item Every $(1, 1)$-transition lies on a 1-cycle or 3-cycle of $(1, 1)$-transitions.
\item When transducing $S$ to $S'$:
\begin{itemize}
\item $K-1$ distinct $(0, 0)$-transitions are used.
\item $K-1$ distinct $(1, 1)$-transitions are used.
\item 1 distinct $(0, 1)$-transition is used.
\item 1 distinct $(1, 2)$-transition is used.
\item Every $(0, 1)$-transition traversal is followed immediately by a $(1, 2)$-transition traversal.
\item The order that transitions are traversed is given as part of the input.
\end{itemize}
\item The first and last $K+2$ symbols of $S$ are $10^K1$. 
\item The first and last $K+2$ symbols of $S'$ are $20^{K-1}12$. 
\end{itemize}
\end{problem}

\begin{lemma}
\label{lem:mmefst-NP-hard}
The \mmefst{} problem is \NP-hard.
\end{lemma}

\begin{proof}
The proof uses a modified version of the reduction from \tpart{} to \mefst{} (Lemma~\ref{lem:promise-min-state-FST-NP-hard}).
A new state $s_{\rm new}$ is added to the solution transducers with the following transitions as seen in Figure~\ref{fig:transducer}:
\begin{itemize}
\item A $(0, 0)$-transition from $s_{\rm new}$ to $s_1$.
\item $(1, 2)$-transition from $s_{\rm new}$ to itself. 
\item A $(0, 1)$-transition from $s_K$ to $s_{\rm new}$, replacing the $(0, 1)$-transition from $s_K$ to $s_1$.
\end{itemize}

\begin{figure}[ht]
\centering
\includegraphics[scale=1.0]{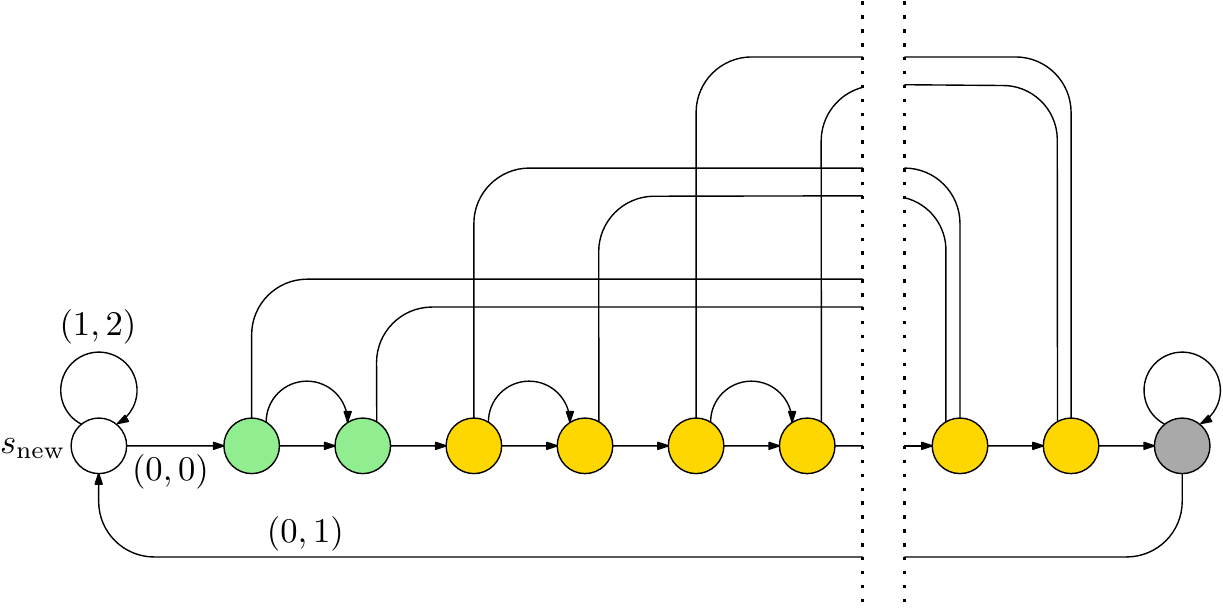}
\caption{The modified min-state FST instances used in the reduction of Theorem~\ref{thm:2xn-U-4PATS-NP-hard}.} 
\label{fig:transducer}
\end{figure}

If the number of states in this transducer is a multiple of~3, add another fixed singleton state to the transducer.
Reassign $K$ to be the number of states in the resulting transducer and label the states $s_1$ to $s_K$ in the order the are traversed by $(0, 0)$-transitions, starting with the state with a $(1, 2)$-transition to itself (formerly called $s_{\rm new}$).

To enforce such solution transducers, replace a segment $0 \rightarrow 1$ in the strings $S$ and $S'$ output by the reduction with $01 \rightarrow 12$.
Also add the segment $10^K1 \rightarrow 20^{K-1}12$ to the beginning of the strings to satisfy the last two constraints of the FST instances in the problem formulation.
The resulting modified strings $S$, $S'$ and integer $K$ yield an instance of \mmefst{} that is ``Yes'' if and only if the \tpart{} instance was also ``Yes''. 
\end{proof}

We expand the reduction above further towards the uniform height-2, 3-PATS problem now. 

\begin{theorem}
\label{thm:2xn-U-4PATS-NP-hard}
The uniform height-2, 3-PATS problem is \NP-hard. 
\end{theorem}

\begin{proof}
Let $(S, S', K)$ be a triple of strings $S \in \{0, 1\}^*, S' \in \{0, 1, 2\}^*$ and an integer $K$ reduced from an instance of \textsc{3-Partition} in the previous proof. 
We reduce it to an instance $(P, |S'|+2K+2)$ of the uniform height-2, 3-PATS problem, where $P$ is the following width-($1{+}|S'|{+}K^2$), height-2 pattern over 3 colors $\{\textcolor{cyan}{\blacksquare}, \textcolor{gray}{\blacksquare}, \textcolor{orange}{\blacksquare}\}$: 
\[\arraycolsep=1.4pt
\begin{array}{ccrrrccr}
\textcolor{cyan}{\blacksquare} & 
\phi(S') & 
\textcolor{cyan}{\blacksquare} \cdots \textcolor{cyan}{\blacksquare} \, \textcolor{gray}{\blacksquare}\textcolor{gray}{\blacksquare} \textcolor{gray}{\blacksquare} & 
w_1 \ \ \ \textcolor{gray}{\blacksquare}\textcolor{gray}{\blacksquare} \textcolor{gray}{\blacksquare} & 
w_2 \ \ \ \textcolor{gray}{\blacksquare}\textcolor{gray}{\blacksquare} \textcolor{gray}{\blacksquare} & 
\cdots & 
\textcolor{gray}{\blacksquare} &
w_{K-1} \ \textcolor{orange}{\blacksquare}\textcolor{orange}{\blacksquare}\textcolor{orange}{\blacksquare} \\

\textcolor{orange}{\blacksquare} & 
\underbrace{\textcolor{orange}{\blacksquare} \textcolor{orange}{\blacksquare} \cdots \textcolor{orange}{\blacksquare}}_{|S'|} & 
\underbrace{\textcolor{cyan}{\blacksquare} \cdots \textcolor{cyan}{\blacksquare}}_{K-3}\textcolor{cyan}{\blacksquare}\textcolor{cyan}{\blacksquare} \textcolor{gray}{\blacksquare} & 
\underbrace{\textcolor{cyan}{\blacksquare} \cdots \textcolor{cyan}{\blacksquare}}_{K-3}\textcolor{cyan}{\blacksquare}\textcolor{cyan}{\blacksquare} \textcolor{gray}{\blacksquare} & 
\underbrace{\textcolor{cyan}{\blacksquare} \cdots \textcolor{cyan}{\blacksquare}}_{K-3}\textcolor{cyan}{\blacksquare}\textcolor{cyan}{\blacksquare} \textcolor{gray}{\blacksquare} & 
\cdots & 
\textcolor{gray}{\blacksquare} &
\underbrace{\textcolor{cyan}{\blacksquare} \cdots \textcolor{cyan}{\blacksquare}}_{K-3}\textcolor{cyan}{\blacksquare}\textcolor{cyan}{\blacksquare} \textcolor{gray}{\blacksquare}
\end{array}
\]
where $\phi: \{0, 1, 2\}^* \to \{\textcolor{cyan}{\blacksquare}, \textcolor{gray}{\blacksquare}, \textcolor{orange}{\blacksquare}\}^*$ is a homomorphism that maps 0 to $\textcolor{cyan}{\blacksquare}$, 1 to $\textcolor{gray}{\blacksquare}$, and 2 to $\textcolor{orange}{\blacksquare}$, respectively, and for $1 \le i \le K{-}1$, 
\[
	w_i = 
	\begin{cases} 
	\textcolor{cyan}{\blacksquare}^{(3i \bmod{K})-1} \textcolor{gray}{\blacksquare} \textcolor{cyan}{\blacksquare}^{K-3-(3i \bmod{K})} & \text{if $3i \le K{-}3 \pmod{K}$} \\
	\textcolor{cyan}{\blacksquare}^{K-3} & \text{otherwise}.
	\end{cases}
\] 
Notice that, for any $1 \le i < j \le K$, $w_i$ and $w_j$ differ in the location of their singleton gray tiles.
Split the pattern $P$ into the leftmost $|S'|{+}1$ columns and the remainder, called the \emph{transduction} and \emph{FST-constructor} gadgets, respectively.
The FST-constructor gadget is further partitioned into $K$ rectangular width-$K$, height-2 blocks.

Next, consider the constraints on RTASs with at most $|S'|{+}2K{+}2$ tile types that uniquely self-assemble $P$.
Lemma~1 of G\"{o}\"{o}s and Orponen~\cite{Goos-2010a} states that any smallest RTAS that uniquely self-assembles a pattern is directed.
As we will prove, directed RTASs uniquely self-assembling $P$ have size at least $|S'|{+}2K{+}2$ tile types; thus we need only consider directed systems.

Let the north and east glues of the seed be~0.
The leftmost $|S'|{+}1$ locations in the bottom row of $P$ are orange, with a cyan location following.
So these positions must be tiled with orange tiles of pairwise-distinct type; the need for $|S'|+1$ distinct orange tile types thus arises.
Similarly, the leftmost $K{-}1$ cyan locations in the bottom row must use $K{-}1$ distinct cyan tile types.
These tile types share the south glue~0, and since the system is directed, their west glues are pairwise distinct.
Label these $K{-}1$ cyan tile types left-to-right $t_{00, 2}, t_{00, 3}, \dots, t_{00, K}$ and the gray tile type immediately right $t_{01}$, as seen below.\footnote{In these later labels, the first subscript indicates the kind of transition of the FST that the tile type will be shown to simulate, e.g., $t_{00, i}$ is a $(0, 0)$-transition, $t_{01}$ and $(0, 1)$-transition, etc.} 
The cyan tile in the northwest corner of $P$ cannot have the same type as any of these $K{-}1$ types, since otherwise this tile can also appear in the southwest corner of $P$.
Call this type $t_0$.
There are $K$ tile types to be colored yet (illustrated as a dotted square). 

\scalebox{0.65}{\begin{tikzpicture}
\Tzero{-2.25}{2}{}{0}{Not 0}{}; \node at (-2.25, 2.75) {$t_0$};

\Tzero{0}{0}{}{$s_1$}{0}{$s_2$}; \node at (0, -1.25) {$t_{00, 2}$};
\Tzero{2.25}{0}{}{$s_2$}{0}{$s_3$}; \node at (2.25, -1.25) {$t_{00, 3}$};
\node at (4, 0) {\Large $\cdots$};
\Tzero{6}{0}{}{$s_{K-3}$}{0}{$s_{K-2}$}; \node at (6, -1.25) {$t_{00, K-2}$};
\Tzero{9}{0}{}{$s_{K-2}$}{0}{$s_{K-1}$}; \node at (9, -1.25) {$t_{00, K-1}$};
\Tzero{12}{0}{}{$s_{K-1}$}{0}{$s_{K}$}; \node at (12, -1.25) {$t_{00, K}$};
\Tone{14.25}{0}{}{$s_K$}{0}{}; \node at (14.25, -1.25) {$t_{01}$};

\foreach \x in {0, 2.25, 6, 9, 12, 14.25} {
\draw[dotted] (\x, 2)++(180:0.5)++(90:0.5) -- ++(0:1) -- ++(270:1) -- ++(180:1) -- ++(90:1);
\node at (\x, 2) {?};
}
\end{tikzpicture}}

These $K$ tile types will turn out to be necessary, implying $(|S'|{+}1){+}(K{-}1)+2 + K{-}1{+}1=|S'|{+}2K{+}2$ types total with $K{-}1$ colored gray and one colored orange. 
For this, we claim that the bottom row of all blocks but the first assemble identically by establishing that the gray tiles attaching to the southeast corner of the first two blocks are identical. 
Suppose not. 
Then the bottom row of the second block cannot reuse cyan tile types used in the bottom row of the first block. 
So the uncolored $K$ tile types must consist of one gray and $K{-}1$ cyan types with south glue~0. 
Thus the complete tile set includes only two gray tile types both of which have the south glue~0. 

Consider the gray tile attaching at the northeast corner of the first block. 
Its south glue is~0 and its west glue is equal to the east glue of the gray tile attaching to its immediate left. 
This contradicts the directedness of the system, since a cyan tile is provided with the same pair of west and south glues. 
Indeed, both gray tile types appear at the southeast corner of a block and to their east are cyan tiles attaching. 

The verified claim brings following properties for all but the first block:
\begin{enumerate}[Property~1:]
\item\label{prop:common_south} For any $1 \le i \le K{-}1$, tiles attaching at the $i$-th top-row position of any two blocks but the first one have the same south glue; tiles attaching at the $K$-th top-row position (northeast corner) of any two blocks including the first one have the same south glue. 
\item\label{prop:distinct_type} Any such pair of tiles have pairwise-distinct east glues (and types).
\item\label{prop:non0_south} The assembly of the bottom row is provided with at least two different kinds of north glues. 
\end{enumerate}
Property~\ref{prop:distinct_type} holds since an orange tile is placed in the northeast corner of only the last block.
Thus without Property~\ref{prop:non0_south}, $\Omega(K^2)$ tile types would be necessary to place the orange tile. 
Observe that for each $1 \le i \le K{-}3$, the $i$-th position of exactly one block is gray and the counterpart of all other blocks are cyan; for each $K{-}2 \le i \le K$, the $i$-th position of only the last block is orange and the counterpart of all others are gray. 
Thus, Properties~\ref{prop:common_south} and~\ref{prop:distinct_type} imply that the tile type set must contain one orange and $K{-}1$ gray tile types whose south glue is equal to the north glue of $t_{01}$ and one gray and $K{-}2$ cyan tile types with a common south glue. 

We claim these requirements enforce that the north glue of $t_{01}$ is not~0.
Suppose otherwise, then the former requirement implies $K{-}2$ extra gray tile types with south glue~0. 
So at most~3 tile types, including $t_0$, have the non-0 south glue, and Property~\ref{prop:non0_south} cannot be satisfied. 
Thus, the north glue of $t_{01}$ is not 0; call it~1.
Tiles attaching at the northeast corner of the blocks must all have distinct types due to Property~\ref{prop:distinct_type}, and now also their south glues must be~1.
The $K$ uncolored tile types thus have south glue~1, and one is colored orange and all the others are colored gray.  

\scalebox{0.65}{\begin{tikzpicture}
\Tzero{-2.25}{2}{}{0}{Not 0}{};\node at (-2.25, 3) {$t_0$};
\Tone{0}{2}{}{}{1}{}; \node at (0, 3) {$t_{11, 2}$};
\Tone{2.25}{2}{}{}{1}{}; \node at (2.25, 3) {$t_{11, 3}$};
\node at (4, 2) {\Large $\cdots$};
\Tone{6}{2}{}{}{1}{}; \node at (6, 3) {$t_{11, K-2}$};
\Tone{9}{2}{}{}{1}{}; \node at (9, 3) {$t_{11, K-1}$};
\Tone{12}{2}{}{}{1}{}; \node at (12, 3) {$t_{11, K}$};
\Ttwo{14.25}{2}{}{}{1}{}; \node at (14.25, 3) {$t_2$};

\Tzero{0}{0}{}{$s_1$}{0}{$s_2$}; \node at (0, -1.25) {$t_{00, 2}$};
\Tzero{2.25}{0}{}{$s_2$}{0}{$s_3$}; \node at (2.25, -1.25) {$t_{00, 3}$};
\node at (4, 0) {\Large $\cdots$};
\Tzero{6}{0}{}{$s_{K-3}$}{0}{$s_{K-2}$}; \node at (6, -1.25) {$t_{00, K-2}$};
\Tzero{9}{0}{}{$s_{K-2}$}{0}{$s_{K-1}$}; \node at (9, -1.25) {$t_{00, K-1}$};
\Tzero{12}{0}{}{$s_{K-1}$}{0}{$s_{K}$}; \node at (12, -1.25) {$t_{00, K}$};
\Tone{14.25}{0}{1}{$s_K$}{0}{}; \node at (14.25, -1.25) {$t_{01}$};
\end{tikzpicture}}

\noindent
Note that a $t_0$ tile cannot attach anywhere in the blocks. 
Indeed, it causes glue mismatch with the seed being placed on the bottom row, and in order for it to attach on the top row, it must share its south glue with $K{-}3$ cyan tile types due to Properties~\ref{prop:common_south} and \ref{prop:distinct_type}. 
In summary, any minimum tile set uniquely assembling $P$ consists of $K$ cyan tile types, $K$ gray ones, and $|S'|{+}2$ orange ones.

Now we prove constraints on the glues of these types. 
With only $K{-}1$ cyan tile types with south glue 0, even the first block must assemble its bottom row as other blocks do. 
That is, the bottom row of all blocks assemble as $t_{00, 2} t_{00, 3} \cdots t_{00, K} t_{01}$. 
Thus, the east glue of $t_{01}$ is equal to the west glue of $t_{00,2}$, that is, $s_1$.
Since $t_0$ does not appear in any block, Property~\ref{prop:distinct_type} implies that the north glues of $t_{00, 2}, t_{00, 3}, \ldots, t_{00, K-2}$ are~0 and that the north glues of $t_{00, K-1}$ and $t_{00, K}$ are~1. 

The top row of the last block is $w_K \textcolor{orange}{\blacksquare} \textcolor{orange}{\blacksquare} \textcolor{orange}{\blacksquare} = \textcolor{cyan}{\blacksquare}^{K-4} \textcolor{gray}{\blacksquare} \textcolor{orange}{\blacksquare} \textcolor{orange}{\blacksquare} \textcolor{orange}{\blacksquare}$. 
The corresponding bottom row exposes the glues $0^{K-3}111$ northwards because it is assembled as $t_{00, 2} \cdots t_{00, K} t_{01}$. 
Since the east glue of all orange tiles but $t_2$ is 0, the last four positions $\textcolor{gray}{\blacksquare} \textcolor{orange}{\blacksquare} \textcolor{orange}{\blacksquare} \textcolor{orange}{\blacksquare}$ of the top row are assembled as $t_{01} t_2 t_2 t_2$. 
This imposes that both the east and west glues of $t_2$ must be equal to the east glue of $t_{01}$, that is, $s_1$. 
Since $S'$ begins with~2, the east glue of $t_0$ is $s_1$. 

\scalebox{0.65}{\begin{tikzpicture}
\Tzero{-2.25}{2}{}{0}{Not 0}{$s_1$};\node at (-2.25, 3) {$t_0$};
\Tone{0}{2}{}{}{1}{}; \node at (0, 3) {$t_{11, 2}$};
\Tone{2.25}{2}{}{}{1}{}; \node at (2.25, 3) {$t_{11, 3}$};
\node at (4, 2) {\Large $\cdots$};
\Tone{6}{2}{}{}{1}{}; \node at (6, 3) {$t_{11, K-2}$};
\Tone{9}{2}{}{}{1}{}; \node at (9, 3) {$t_{11, K-1}$};
\Tone{12}{2}{}{}{1}{}; \node at (12, 3) {$t_{11, K}$};
\Ttwo{14.25}{2}{}{$s_1$}{1}{$s_1$}; \node at (14.25, 3) {$t_2$};

\Tzero{0}{0}{0}{$s_1$}{0}{$s_2$}; \node at (0, -1.25) {$t_{00, 2}$};
\Tzero{2.25}{0}{0}{$s_2$}{0}{$s_3$}; \node at (2.25, -1.25) {$t_{00, 3}$};
\node at (4, 0) {\Large $\cdots$};
\Tzero{6}{0}{0}{$s_{K-3}$}{0}{$s_{K-2}$}; \node at (6, -1.25) {$t_{00, K-2}$};
\Tzero{9}{0}{1}{$s_{K-2}$}{0}{$s_{K-1}$}; \node at (9, -1.25) {$t_{00, K-1}$};
\Tzero{12}{0}{1}{$s_{K-1}$}{0}{$s_{K}$}; \node at (12, -1.25) {$t_{00, K}$};
\Tone{14.25}{0}{1}{$s_K$}{0}{$s_1$}; \node at (14.25, -1.25) {$t_{01}$};
\end{tikzpicture}}

Since $S'$ ends with~2, no tile (necessarily of type $t_{11, 2}, t_{11, 3}, \dots, t_{11, K}$ by Properties~\ref{prop:common_south} and~\ref{prop:distinct_type}) appearing at the northeast corner of a block has east glue $s_1$. 
Moreover, tiles attaching to their east are of type $t_{00, 2}, \ldots, t_{00, K}$ or $t_{01}$, thus their east glues are in $\{s_2, s_3, \ldots, s_K\}$.
Without loss of generality, assign them as follows: 

\scalebox{0.65}{\begin{tikzpicture}
\Tzero{-2.25}{2}{}{0}{Not 0}{$s_1$};\node at (-2.25, 3) {$t_0$};
\Tone{0}{2}{}{}{1}{$s_2$}; \node at (0, 3) {$t_{11, 2}$};
\Tone{2.25}{2}{}{}{1}{$s_3$}; \node at (2.25, 3) {$t_{11, 3}$};
\node at (4, 2) {\Large $\cdots$};
\Tone{6}{2}{}{}{1}{$s_{K-2}$}; \node at (6, 3) {$t_{11, K-2}$};
\Tone{9}{2}{}{}{1}{$s_{K-1}$}; \node at (9, 3) {$t_{11, K-1}$};
\Tone{12}{2}{}{}{1}{$s_K$}; \node at (12, 3) {$t_{11, K}$};
\Ttwo{14.25}{2}{}{$s_1$}{1}{$s_1$}; \node at (14.25, 3) {$t_2$};

\Tzero{0}{0}{0}{$s_1$}{0}{$s_2$}; \node at (0, -1.25) {$t_{00, 2}$};
\Tzero{2.25}{0}{0}{$s_2$}{0}{$s_3$}; \node at (2.25, -1.25) {$t_{00, 3}$};
\node at (4, 0) {\Large $\cdots$};
\Tzero{6}{0}{0}{$s_{K-3}$}{0}{$s_{K-2}$}; \node at (6, -1.25) {$t_{00, K-2}$};
\Tzero{9}{0}{1}{$s_{K-2}$}{0}{$s_{K-1}$}; \node at (9, -1.25) {$t_{00, K-1}$};
\Tzero{12}{0}{1}{$s_{K-1}$}{0}{$s_{K}$}; \node at (12, -1.25) {$t_{00, K}$};
\Tone{14.25}{0}{1}{$s_K$}{0}{$s_1$}; \node at (14.25, -1.25) {$t_{01}$};
\end{tikzpicture}}

So the east glues of all tile types in the FST-construction gadget are in $\{s_1, \dots, s_K\}$. 
The east glues of $t_{11, 2}, \dots, t_{11, K}$ are distinct and selected from $\{s_1, \dots, s_K\}$.
Since $t_{11, 2}, \dots, t_{11, K}$ share the south glue~1 with $t_2$, the west glue of $t_2$ is $s_1$, and the system is directed, the west glues are distinct and from $\{s_2, s_3, \ldots, s_K\}$.

The glue~0 is not in $\{s_1, \ldots, s_K\}$, as otherwise a cyan or gray tile could appear in the southwest corner of $P$.
So none of the cyan, gray, or $t_2$ tile types has east glue~0 and thus a $t_0$ tile cannot attach anywhere but the northwest corner of $P$. 
In order to assemble the rest of the top row, none of the orange tile types but $t_2$ can be used. 
These orange tile types share the south glue~0 with $t_{00, 2}, t_{00, 3}, \ldots, t_{00, K}, t_{01}$ so that for the sake of directedness, their west glue cannot be taken from $\{s_1, s_2, \ldots, s_K\}$ but the east glue of other tile types, which are illustrated above, is taken from this set. 
In summary, the tile types usable to assemble the top row except its leftmost position are all the cyan tile types, all the gray ones, and the orange tile type $t_2$. 
We can see that these tile types simulate a transition of an FST with $K$ states $s_1, s_2, \ldots, s_K$ by interpreting their west, south, and east glues as the source of transition, a letter read, and the target of transition, and their color as an output (via $\{\textcolor{cyan}{\blacksquare}, \textcolor{gray}{\blacksquare}, \textcolor{orange}{\blacksquare}\} \rightarrow \{0, 1, 2\}$).

Until the west glues $\{s_2, s_3, \ldots, s_K\}$ are assigned to $t_{11, 2}, t_{11, 3}, \ldots, t_{11, K}$, the simulated FST is not determined.
Nevertheless, the other tiles imply the following properties of the FST: 
\begin{itemize}
\item The FST has a $K$-cycle of $K-1$ (0, 0)-transitions and one (0, 1)-transition starting at the initial state $s_1$. 
\item The state $s_1$ has a (1, 2)-loop. 
\item When making a transition from $s_1$ to another state, 1 is never output. 
\item The only transition to $s_1$ with output 1 is the (0, 1)-transition from $s_K$. 
\end{itemize} 

Now we will prove that the FST must be the one in Figure~\ref{fig:transducer}. 
First, we will prove that the north glues of the orange tiles below $\phi(S')$ in $P$ must be $S$.
Observe that the tile types imply several constraints on the north glues exposed:
\begin{itemize} 
\item A~1 north glue cannot be below a cyan tile (no cyan tile has south glue~1). 
\item A~0 north glue cannot be below an orange tile (no orange tile has south glue~0). 
\item Among consecutive gray tiles, only the eastmost can have a~0 south glue (any gray tile with south glue~0 is not followed by another gray tile, since 1~is never output at $s_1$).
\item The eastmost gray tile in a sequence of consecutive gray tiles followed by an orange tile must have south glue~0 due to the (0, 1)-transition ($t_{01}$ is the unique gray tile type with east glue $s_1$ and has south glue~0). 
\end{itemize}

The only remaining scenario to check is consecutive gray tiles followed by a cyan tile. 
Recall that $S$ and $S'$ consist of three kinds of segments: fixed singletons, half-fixed triples, and half-fixed intervals, and each such segment is preceded by the segment $01 \rightarrow 12$ at the end of the previous such segment (or added initial segment $10^K1 \rightarrow 20^{K-1}12$).  
The output (``$S'$-part'') of a fixed singleton has the form $0^{i-1}10^{K-i}12$. 
The first~1 cannot be output by the (0, 1)-transition (i.e., cannot be tile type $t_{01}$ and must have south glue~1) as otherwise, the transduction does not reach the state $s_K$ after processing the following~0's and~12 is not subsequently output.
Thus, the input (``$S$-part'') of a fixed singleton with output $0^{i-1}10^{K-i}12$ must be $0^{i-1}10^{K-i}01$ and in order to assemble $\phi(S')$ in the pattern $P$, the sequence of north glues of the orange tiles below must be $0^{i-1}10^{K-i}01$.

The same argument implies that in order to assemble a portion of $\phi(S')$ encoding the output of a half-fixed triple, the sequence of north glues of orange tiles must encode the corresponding input of a half-fixed triple.
Because $S$ and $S'$ always contain many fixed singleton segments, tile type $t_{11, K}$ has the west glue~$s_K$; moreover, no other tile type $t_{11, i}$ has west glue~$s_K$, i.e., the FST has no other transition to $s_K$ with output~1.

Finally, we analyze the half-fixed interval, whose output is of the form $(0^{i-1}11010^{K-i-2}1)(20^{i+1}11010^{K-i-4}1)$. 
The~1 in the second occurrence of 10 cannot be output by the (0, 1)-transition, by the previous argument (the transduction does not reach state $s_K$).
In order for the~1 in the first occurrence of 10 to be output by the (0, 1)-transition, the preceding 1 must be the output of the (1, 1)-loop at $s_K$. 
However, there are insufficient preceeding 0s ($i \leq \lceil 2K/3 \rceil$ always) following the previous~2 (at the end of the previous segment) to have reached state $s_K$.
Thus the~1's in both occurrences of 10 in the first parenthesis are not output by the (0, 1)-transition, and instead by (1, 1)-transitions (i.e., are not placements of tile type $t_{01}$ and so have south glues~1). 
The same arguments applies to occurrences in the second parenthesis.

The only remaining flexibility in the design of the tile set is assigning west glues to $t_{11, 2}, \dots, t_{11, K-1}$.
The fixed singletons, half-fixed triples, and half-fixed intervals, whose inputs are a sequence of (hidden) north glues of the bottom (orange) row of the transduction gadget that encode $S$, allow for only the intended assignment illustrated in Figure~\ref{fig:transducer}. 
Consequently, the assignment is possible if and only if the given 3-partition has a solution. 
\end{proof}

\section{Efficiently Solvable \pats{} Problems}
\label{sec:PATS-easy}

The non-uniform height-1 PATS problem is trivially solvable using one tile type for each color.
This idea can be generalized for all patterns of fixed height: 

\begin{theorem}
\label{thm:hxn-NU-cPATS-linear}
The non-uniform height-$h$ $c$-\pats{} problem can be solved in $c^{c^{O(h)}}n$ time.
\end{theorem}

\begin{proof}
Observe that the minimum-size tile set that assembles any given pattern has size at most $T = hc^h$, since there are $c^h$ possible columns and a distinct set of $h$ tiles for each column (whose appearance is programmed by the seed) can be used to assemble the pattern.
The algorithm is a brute-force search for a smallest solution tile set (of size at most $T$), using dynamic programming to check each tile set in $O(hT^{h+2}n)$ time.
Since a tile can be specified by a binary string of length $4\log{T}\log{c}$, there are at most $2^{4T\log{cT}} = (cT)^{4T}$ such tile sets.
Thus the algorithm runs in $O((cT)^{4T}hT^{h+2}n) = T^{O(T)}n = c^{c^{O(h)}}n$ time.

All that remains is to describe the dynamic programming algorithm.
Lemma~1 of G\"{o}\"{o}s and Orponen~\cite{Goos-2010a} states that any smallest tile set is directed.
The subproblems solved have the form: ``Does the tile set deterministically assemble the first $i$ columns of the pattern with top-to-bottom sequence of tile types $t_1, t_2, \dots, t_h$ in column $i$?''

All subproblems for column $i$ can be solved by checking each of at most $T \cdot T^h = T^{h+1}$ combinations of north seed glue in column $i$ and ``Yes'' subproblems for column $i-1$, and recording ``Yes'' for the top-to-bottom tile sequences resulting from deterministic assembly of column $i$ of the pattern and ``No'' for all other sequences.
For each column, computing the ``Yes'' sequences takes $O(T^{h+1} \cdot hT)$ time and recording the solutions to all subproblems takes $O(T^h)$.
So across all columns the algorithm takes $O(hT^{h+2}n)$ time. 
\end{proof}

As established in Section~\ref{sec:cPATS-hard}, a similar algorithm for the uniform model is impossible unless $\P = \NP$.
Nevertheless, the uniform height-1 PATS problem can be solved in linear time using a pigeonhole argument and a DFS-based search for the longest repetitive suffix of a given height-1 pattern:

\begin{theorem}
\label{thm:1xn-U-PATS-linear}
The uniform height-1 \pats{} problem can be solved in $O(n)$ time.
\end{theorem}

\begin{proof}
Consider an input width-$n$ height-1 pattern $P$ for the uniform height-1 PATS problem as a word of length $n$ over an alphabet of colors $C$.
Let $x, y$ be distinct suffices of $P$ with $y$ a prefix of $x$ and $y$ as long as possible.
Let $x = zy$ for some nonempty word $z$. 
Since $y$ is a prefix of $x$, $z$ is a period of $x$ and thus $x = z^i z_p$ for some $i \ge 1$ and prefix $z_p$ of $z$. 
Moreover, $z$ is primitive (not a power of another word) since $y$ is as long as possible.

Let $m = n{-}|y|$.
We prove that $m$ tile types are necessary and sufficient for a uniform RTAS to uniquely self-assemble $P$. 
For sufficiency, use a tile set that hardcodes the prefix of $P$ preceeding $x$ with $m{-}|x|$ tile types and uses a repeating set of $|x|{-}|y|$ tile types to assemble the repetitions of $z$. 
For necessity, suppose $P$ can be uniquely self-assembled using strictly less than $m$ tile types.
By the pigeonhole principle, there must exist $1 \le i < j \le m$ such that the tiles at $(i, 1)$ and $(j, 1)$ have the same type.
Then the suffixes $x'$ and $y'$ of $P$ starting at positions $(i, 1)$ and $(j, 1)$ are distinct suffices of $P$ with $y' > y$, a contradiction with the previous choice of $x$ and $y$.
So $m$ is the minimum number of tile types to uniquely assemble $P$. 

All that remains is to prove that $y$ (and $x$) can be computed in $O(n)$ time.
This can be done by DFS in a suffix tree, searching for the longest suffix ending at a non-leaf node.
The suffix tree can be constructed from $P$ in $O(n)$ time~\cite{Farach-1997}, and the DFS operates in the same running time.
\end{proof}

\section{Conclusion}

Our work here extends the extensive prior work on the parameterized $c$-\pats{} problem to also incorporate pattern height and uniformity, and finds a more delicate complexity landscape: limited height and colors do not make the \pats{} problem tractable, except when combined in the non-uniform model, or in degenerate cases (height-1 or $1$-\pats{}).
A single combination of parameters and model remains unresolved; we conjecture the following:

\begin{conjecture}
The uniform height-2 2-\pats{} problem is \NP-hard. 
\end{conjecture}

We encourage further parameterized analysis of problems in tile self-assembly in support of recent efforts in developing a more complete understanding of the structural complexity of tile self-assembly (see~\cite{Woods-2015a}). 

\section*{Acknowledgements} 
We thank Yo-Sub Han for very fruitful discussions about finite automata and tile self-assembly, and anonymous reviewers for comments that improved the paper.

\bibliographystyle{plain}
\bibliography{thin_pats}

\end{document}